\documentclass{svjour2}

\usepackage{amssymb}
\usepackage{comment}
\usepackage{amsfonts}
\usepackage{amsmath}
\usepackage{graphicx}

\def\sE{\mathcal{E}}
\def\sH{\mathcal{H}}
\def\E{\mathbb{E}}
\def\P{\mathbb{P}}

\newcommand{\N}{{\mathbb N}}
\newcommand{\PP}{{\mathbb P}}

\newcommand{\EE}{{\mathbb E}}
\newcommand{\FF}{{\mathcal F}}
\newcommand{\FFF}{{\mathbb F}}
\newcommand{\GGG}{{\mathbb G}}

\newcommand{\EN}{{\mathcal E}}

\newcommand{\prfo}[1]{ ( #1 )_{t\in [0,1]}} 

\journalname{Annals of Finance}

\begin{document}

\title{On the semimartingale property via bounded
logarithmic utility
}


\author{Kasper Larsen         \and
        Gordan \v Zitkovi\' c
}


\institute{Kasper Larsen \at
              Department of Mathematical Sciences, Carnegie Mellon
University, Pittsburgh, PA 15213, USA\\
              \email{kasperl@andrew.cmu.edu}           
           \and
           Gordan \v Zitkovi\' c \at
              Department of Mathematics, University of Texas at
Austin, Austin, TX 78712, USA \\\email{gordanz@math.utexas.edu} }

\date{Received: date / Accepted: date}

\maketitle

\begin{abstract}
This paper provides a new version of the condition of Di Nunno et
al. (2003), Ankirchner and Imkeller (2005) and Biagini and \O
ksendal (2005) ensuring the semimartingale property for a large
class of continuous stochastic processes. Unlike our predecessors,
we base our modeling framework on the concept of portfolio
proportions which yields a short self-contained proof of the main
theorem, as well as a counterexample, showing that analogues of
our results do not hold
in the discontinuous setting.\ \\

\ \\

\end{abstract}

\section{Introduction and summary}

In \cite{DelSch94} the connection between the concept of
no-arbitrage and the assumption that the financial assets are driven
by semimartingales is initiated. Here it is shown there that if the
financial market satisfies the condition
 of \emph{no free lunch with vanishing risk} for simple
trading strategies, then the traded
securities allow for a semimartingale decomposition.
This and similar results depend heavily on the mathematical constructs
- the theory of stochastic integration employed, and
the class of the integrands
used - to describe the economic concept of no-arbitrage.
\cite {BjoHul05} illustrates possible pitfalls resulting from
an attempt of economic interpretation of mathematical results
based on an integration theory at odds with the financial intuition.

The main result of the present paper is inspired by
\cite{BiaOks05} and \cite{DinMeyOksPro03}, and states (loosely)
that a continuous process with finite quadratic variation is a
semimartingale if the expected utility of a logarithmic investor
is uniformly bounded from above over a specific natural class of
trading strategies. Unlike \cite{BiaOks05} and
\cite{DinMeyOksPro03}, we do not replace the It\^{o} integration
with the anticipative {\em forward} integration, and we do not
assume the existence of a trading strategy that achieves the
optimal expected logarithmic utility. In fact, the existence of
such a strategy is one of the conclusions of our main theorem.

The recent and independent paper \cite{AI05} develops an idea
similar to ours and relates the semimartingality of the
stock-price process to the boundedness of the expected utility.
The authors base their approach on simple buy-and-hold strategies
thereby circumventing the lack of a stochastic integration theory.
Our approach is different - it hinges on the observation that a
canonical integration theory can be based on simple portfolio {\em
proportions}, without falling into the traps described in
\cite{BjoHul05}. Indeed, while one of the main results of
\cite{AI05} is that bounded utility implies semimartingality,
regardless of the continuity properties of the process under
scrutiny, our epilogue is different. We construct an example of a
discontinuous non-semimartingale $S$ with the property that the
expected logarithmic utility is uniformly bounded over all
strategies in which portfolio proportions are simple processes.
Moreover, there exists a shrinkage of the original filtration
under which the process $S$ is a semimartingale. The existence of
such a counterexample poses the following question: {\em Can we
describe (and work with) a class of stochastic processes, strictly
larger than the class of semimartingales, for which the
logarithmic investors will not be able to achieve arbitrarily
large expected utilities?} While the non-semimartingales in this
class will surely admit free lunch with vanishing risk, the
possibiliy of their use in financial modeling is not ruled out.
Indeed, the logarithmic investors will not demand unlimited
quantities of such securities. We leave this question for further
research.

In the continuous case, the flavor of our results agrees
with \cite{AI05}, but our
approach provides new insights in several respects. First, the proof of
our main theorem is short and self-contained, and uses a simple
Hilbert-space argument. As a consequence of this, we are able
to explicitly derive the semimartingale decomposition of the stock
price in terms of the Riesz representation of a suitably defined linear
functional. The proof of the related result in \cite{AI05} is based
on the already mentioned result of \cite{DelSch94}, and provides only the
abstract existence of the semimartingale decomposition.
Second, a byproduct of our analysis is the
existence of the optimal trading strategy for an investor with logarithmic
utility - the growth-optimal portfolio.

The paper is structured as follows: In the first section we describe
the framework and prove our main result. The second section provides
a counterexample which illustrates the fact that, when jumps are
present,  bounded logarithmic utility on
 simple portfolio proportions is not sufficient to grant the
semimartingality of the price process.

As all our stochastic processes are defined on the time
horizon $[0,1]$, we will consistently use the shorthand $S$ for the
process $\prfo{S_t}$, throughout the paper.

\section{The Main Result for Continuous Processes}

\subsection{The Modeling Framework}
\label{subs:modfra}
We consider a continuous stochastic process $S$,
defined on the unit time horizon $\left[0,1\right]$, and adapted to a
complete and right-continuous filtration
$\FFF\triangleq \prfo{\FF_t}$,
 on some probability space
$(\Omega,\FF,\PP)$. We assume that $S$ has finite quadratic
variation on $\left[0,1\right]$,  meaning for $\omega \in \Omega$
the following limit exists
\begin{equation}
\label{equ:qv}
\lim_{n\to\infty} \sum_{k=1}^{n}
\left(S_{\frac{k}{n}}-S_{\frac{k-1}{n}}\right)^2
\end{equation}
and is finite. In that case, the process $S$ defined by
\begin{align}\label{qv}
\left[S\right]_t\triangleq \lim_{n\to\infty} \sum_{k=1}^{n}
        \left(S_{t\wedge \frac{k}{n}}-S_{t\wedge\frac{k-1}{n}}\right)^2
\end{align}
is finite valued, non-decreasing and continuous. Of course, the
sequence $\{0,\frac{1}{n},$ $\dots,$ $\frac{n-1}{n},1\}$ of partitions in
\eqref{qv} is chosen for simplicity; any other sequence with
comparable properties would lead to the same conclusions.

\begin{remark}
Arguably the most natural way to ensure the existence of the
quadratic variation, as defined in \eqref{equ:qv}, is to assume
the existence of a filtration $\mathbb{F}'\triangleq
\prfo{\mathcal{F}'_t}$ smaller than $\FFF$
(i.e.,~$\mathcal{F}'_t\subseteq \FF_t$, $t\in \left[0,1\right]$)
such that $S$ is an $\mathbb{F}'$-se\-mi\-mar\-tin\-ga\-le.
Indeed, semimartingales have finite quadratic variation, which,
being defined in a pathwise manner, remains undisturbed under
enlargements of the filtration. Inside-trading models typically
assume the existence of two classes of investors - {\em regular
investors}, with access to the public information $\mathbb{F}'$,
and {\em insiders}, whose information set is modeled by $\FFF$,
see the seminal paper \cite{KarPik96}. In this setting, $S$
represents the stock-price (or the return) process described by a
semimartingale from the regular investor's point. Our main result,
below, can be applied in this setting as a sufficient condition on
the insider's (superior) information structure $\FFF$, so that $S$
remains a semimartingale under $\FFF$ as well.
\end{remark}

\subsection{Some Classes of Stochastic Processes}
Let ${\mathcal{H}}^{s}$ denote the set of all stochastic processes
$\pi$ of the form:
\begin{equation}
\pi_{t}=\sum_{i=1}^{n}K_{i}{\mathbf{1}}_{{(T_{i-1},T_{i}]}}%
(t),\label{equ:simple}%
\end{equation}
where $n\in{\mathbb{N}}$, $0\triangleq T_{0}\leq T_{1}\leq\dots\leq
T_{n}\triangleq 1$ are $\FFF$-stopping times and
$K_{i}\in{\mathbb{L}}^{\infty
}(\mathcal{\mathcal{F}}_{T_{i-1}}),~i=1,\ldots,n$.
\begin{remark}
An analogous class of processes with the filtration $\FFF$ replaced by
a different filtration $\GGG$ will be needed in Section 3. Such a
class will be denoted by $\sH^s(\GGG)$.
\end{remark}
It is well known that
\begin{equation}
{\mathcal{H}}^{2}\triangleq\left\{  \pi:\pi\text{ is
predictable and }\,{||\pi||}_{{\mathcal{H}}^{2}}%
<\infty\right\}  ,\nonumber
\end{equation}
is a Hilbert space where ${||\pi||}_{{\mathcal{H}}^{2}}^2
\triangleq{\mathbb{E}}\int_{0}%
^{1}\pi_{u}^{2}\,d\left[S\right]_{u}$. As no integrability
assumptions will be placed on either $S$ or $\left[S\right]$, the
stopping-time sequence $\{T_n\}_{n=1}^\infty$, where
\begin{equation}\label{stoppingtimes}
        T_n\triangleq\inf\{t\leq 1: \vert S_t\vert >n\}\land \inf\{t\leq
        1:\left[S\right]_t>n\}\land 1,
\end{equation}
will prove useful in the reduction arguments in the sequel. Indeed,
$T_n\leq T_{n+1}\leq 1$ for all $n$, and $\P(T_n=1)\to 1$ for
$n\to\infty$. Finally, we define $\sH^{b}
\triangleq\cup_{n\in\mathbb{N}} \sH^{b}_n$, where
\begin{align*}
        \sH^{b}_n \triangleq\{ \pi\in\sH^2\,:
        \,\pi_t=0 \text{, for $t>T_n$} \}, \ n\in\N.
\end{align*}
We are now ready to state and prove  the following auxiliary  result.
\begin{proposition}\label{density}
$\sH^{b}\cap\sH^{s}$ is dense in $\sH^2$ with respect to the
norm  $||\cdot||_{\sH^2}$.
\end{proposition}
\begin{proof}
Pick a $\pi\in\sH^2$, and note that
\begin{equation}
    \nonumber
    \begin{split}
        ||\pi-\pi{\bf 1}_{\left[0,T_n\right]}||^2_{\sH^2}=\E\left[\int_{T_n}^1 \pi_u^2\,
        d\left[S\right]_u\right]\to 0,\text{ as $n\to\infty$}
    \end{split}
\end{equation}
by the Dominated Convergence Theorem. Hence, the family $\sH^{b}$
is dense in $ \sH^2$. To finish the proof of the proposition, it
is enough  to show that $\sH^{b}\cap\sH^{s}$ is dense in
$\sH^{b}$. For this, in turn, it suffices to prove that
$\sH^{b}_n\cap\sH^{s}$ is dense in $\sH^{b}_n$, for each $n$. To
verify this statement, we define the process $A_t \triangleq
[S]^{T_n}_t$ and apply Lemma 2.7, p. 135 in \cite{KarShr91}.
$\lozenge$
\end{proof}
\endproof

\subsection{The Canonical Definition of Stochastic Exponentials}

\label{sub:canonical-definition} Although no It\^o-type
integration theory exists for general adapted integrands with
respect to a process $S$ merely satisfying the assumptions of
Subsection \ref{subs:modfra}, we can always define the stochastic
integral for an integrand $\pi\in\sH^{s}$, of the form
\eqref{equ:simple}, in the familiar way
\begin{equation}
(\pi\cdot S)_{t}\triangleq\sum_{i=1}^{n}
K_{i}\left(S_{T_{i}\wedge t}-S_{T_{i-1}\wedge t}\right).
\label{equ:Brownian-integral}%
\end{equation}
More importantly for our results, stochastic
exponentials can be defined canonically as well by
\begin{equation}
\sE(\pi\cdot S)_{t}\triangleq\exp\left(\int_{0}^{t}\pi_{u}\,dS_{u}%
-\frac{1}{2}\int_{0}^{t}\pi_{u}^{2}\,d\left[S\right]_{u}%
\right)\label{equ:exponential-for-Brown}%
\end{equation}
for all $\pi\in\sH^{s}$ with the $dS$-integral inside the
exponential function defined by (\ref{equ:Brownian-integral}). For
$\pi\in \sH^{s}$, we can show that ${\mathcal{E}}(\pi\cdot S)$ is
the unique pathwise solution $Z$ to the Dol\'{e}ans-Dade stochastic
differential equation
\begin{align}\label{eq:dd}
dZ_t = Z_t \pi_t dS_t,\quad Z_0=1.
\end{align}
Of course, the integrand $Z\pi$ appearing in (\ref{eq:dd}) is
not necessarily in ${\mathcal{H}}^{s}$. Nevertheless, the
integral $\int_0^tZ_{u}\pi_{u}\,dS_{u}$ exists a.s.~as a limit of
Riemann sums and equals $Z_{t}-Z_{0}$. In order to see this,
for a given $\pi\in\sH^{s}$ and a fixed $\omega \in
\Omega$, we define the continuous function $x$ on $\left[0,1\right]$ by
$$
x(t) \triangleq \int_0^t \pi_u dS_u - \frac12\int_0^t\pi^2_ud\left[S\right]_u.
$$
The quadratic variation of this deterministic function is given by
$
\left[x\right](t) = \int_0^t \pi_u^2d\left[S\right]_u.
$
The statement now follows from combining these expressions with
Exercise 3.13 p. 153 in \cite{RY99} applied to the function $F(x)
\triangleq \exp(x)$.

\subsection{A Financial Interpretation}
\label{subs:fin-int} We consider a simple financial market
consisting of two assets: one risk-free asset with a zero interest
rate, and one risky asset whose price at time $t$ will be denoted by
$P_t$ and is given by $P_t\triangleq \EN(S)_t$. For any simple
process $\pi\in\sH^{s}$ of the form \eqref{equ:simple}, the
following equation will be used as a definition of the wealth
process of a financial agent investing in the market
\begin{align}\label{wealthdef}
W^\pi_0 \triangleq 1,\quad W^{\pi}_t \triangleq \sE(\pi\cdot
S)_t\;\;\text{for}\;\;t>0.
\end{align}
We will interpret the value $\pi_t$ as the proportion of his/her
current wealth, the  investor has invested in the risky asset at
time $t$. In order to motivate this terminology, let us assume for a
second that the return process $S$ is a semimartingale, and hence,
$P_t \triangleq \sE(S)_t$ in the classical sense. With the process
$H$ denoting the number of shares of the risky asset in the
investor's self financing portfolio,
 the investor's wealth evolves according to the following equation
$$
dW_t = H_t\, dP_t = H_tP_t\, dS_t =
\left(\frac{H_tP_t}{W_t}\right)W_t\, dS_t.
$$
This translates exactly into $W_t = \sE(\pi\cdot S)_t$, for $\pi$ given
by
$$
\pi_t \triangleq \frac{H_tP_t}{W_t}.
$$
Of course, the process $S$ in our framework is not assumed to be a
semimartingale, so the above discussion cannot be transferred to our
setting directly. However, the discussion in Subsection
\ref{sub:canonical-definition} implies that such a transformation is
indeed feasible, as long as we use only simple processes
$\pi\in\sH^{s}$.
\begin{remark}
A somewhat different interpretation of the equation \eqref{wealthdef}
stems from the alternative assumption that $S$ itself (and not $P=\EN(S)$) is the
price process of the risky asset. In that case, the equation
\eqref{wealthdef} still describes the evolution of the investor's
wealth, but now under that understanding that $\pi_t$ denotes the
number of shares of the risky asset held {\em per unit of wealth} - a
concept less common than that of portfolio proportions described
above.
\end{remark}
Finally, we impose the following assumption on the risk-aversion
characteristics of our
financial agent: his/her goal is to invest in such a way as to
maximize the expected logarithmic utility of the terminal wealth. The
agents with this objective are commonly called log-investors.

\subsection{The Main Result}
In financial terms, the premise of our main result is that the
price process of the risky asset is such that the expected utility
of the $\log$-investor is uniformly bounded over all simple
portfolio-proportion processes $\pi\in\sH^{s}$, i.e.,
\begin{equation}
\sup_{\pi\in{\mathcal{H}}^{s}}{\mathbb{E}}\left[\log(W_{1}^{\pi
})\right]<\infty.\label{equ:finite-log-Brown}
\end{equation}
In this expression, we implicitly use the convention that if
$\pi\in\sH^{s}$ renders both the positive and the negative part of
$\log(W^\pi_1)$ non-integrable, we define
$\E\left[\log(W^\pi_1)\right]\triangleq -\infty$. This convention
is widely used in the theory of utility maximization, see e.g.,
\cite{AI05} p. 482.

Our main result in the continuous setting is the following.
\begin{theorem}
\label{thm:main-Brownian} Let $S$ be a
continuous adapted stochastic process with finite quadratic
variation in the sense of \eqref{equ:qv}, satisfying the condition
\eqref{equ:finite-log-Brown}. Then $S$ is a semimartingale
with decomposition
\begin{align}\label{decomposition}
S_{t}=\hat{S}_{t}+\int_{0}^{t}\alpha_{u}\,d\left[S\right]_{u},
\end{align}
where $\hat{S}$ is a local martingale and $\alpha$
is a predictable process in $\sH^2$.
\end{theorem}

\begin{proof}
For $\pi \in \sH^{s}\cap\sH^{b}$, both integrals
$
\int_{0}^{1}\pi_{u}\,dS_{u}$ and $\int_{0} ^{1}\pi_{u}^{2}\,d\left[S\right]_{u}
$
are  well-defined and have finite expectations; in particular the
linear functional $\Lambda$, defined by
$$
\Lambda(\pi)\triangleq \E\left[\int_0^1\pi_u\, dS_u\right],\  \pi\in\sH^s\cap\sH^b
$$
is well-defined and finite valued on
$\sH^{s}\cap\sH^{b}$.
Assumption \eqref{equ:finite-log-Brown} grants the existence of a finite constant $C$
such that
$$
\E\left[\int_{0}^{1}\pi_{u}\,dS_{u}-\frac{1}{2}\int_{0}%
^{1}\pi_{u}^{2}\,d\left[S\right]_{u}\right] = \E\left[\log(W^{\pi}_1)\right] \le C,\text{
for all $\pi\in\sH^{s}\cap\sH^{b}$ }.
$$
Therefore, $\Lambda$ admits the following bound
\[
\Lambda(\pi)\leq C+\frac{1}{2}{||\pi||}_{{\mathcal{H}}^{2}}^{2},
\]
which can be strengthened by noting that
\begin{equation}
\Lambda(\pi)=\frac{1}{\gamma}\Lambda(\gamma\pi)\leq\frac{C}{\gamma}%
+\frac{\gamma}{2}{||\pi||}_{\mathcal{H}^{2}}^{2}, \text{ for each
  $\gamma>0$.}
\label{equ:gammas}%
\end{equation}
Minimization the right-most part of (\ref{equ:gammas}) with respect to
$\gamma$ yields that
\[
|\Lambda(\pi)|\leq\sqrt{2C}{||\pi||}_{\mathcal{H}^{2}},\text{ for all $\pi\in\sH^{s}\cap\sH^{b}$.}
\]
From this we conclude that $\Lambda$ is a continuous linear
functional on $\sH^{s}\cap\sH^{b}$. Proposition \ref{density}
states that  $\sH^{s}\cap\sH^{b}$ is dense in $\sH^2$ with respect
to the topology induced by the norm
${||\cdot||}_{\mathcal{H}^{2}}$. Consequently,  the linear
functional $\Lambda$ admits a unique linear and continuous
extension to $\sH^2$. Riesz's representation theorem guarantees
the existence of a process $\alpha\in{\mathcal{H}}^{2}$ such that
\begin{equation}
{\mathbb{E}}\left[\int_{0}^{1}\pi_{u}\,dS_{u}\right]=\Lambda(\pi)={\mathbb{E}}\left[\int_{0}%
^{1}\pi_{u}\alpha_{u}\,d\left[S\right]_{u}\right],\text{ for
$\pi\in\sH^{s}\cap\sH^{b}$.}\label{equ:Riesz}%
\end{equation}

The proof is concluded by showing that the sequence
$\{T_n\}_{n\in\N}$ defined by \eqref{stoppingtimes} can be used to
reduce the continuous adapted process
$$
\hat{S}_t \triangleq S_t - \int_0^t\alpha_ud\left[S\right]_u
$$
to a martingale.
To this end, we let $\tau$ be an arbitrary stopping time and define
the simple process $\pi^n$ in $\sH^{s}\cap\sH^{b}$ by
$\pi_u^n\triangleq\mathbf{1}_{\left[0,\tau\land
    T_n\right]}(u)$. Applying the equality \eqref{equ:Riesz} to $\pi^n$ yields
\begin{equation}
{\mathbb{E}}\left[S_{\tau\land T_n}\right] -
S_0={\mathbb{E}}\left[\int_{0}^{\tau\land
T_n}\alpha_{u} \, d\left[S\right]_{u}\right],\text{ that is, }\ {\mathbb{E}}\left[\hat{S}^{T_n}_{\tau}\right]=S_0.\label{equ:expect-is-finite}%
\end{equation}
Since (\ref{equ:expect-is-finite}) holds for all stopping times
$\tau$, it follows that $\hat{S}^{T_n}$ is a martingale. Hence,
$\{\hat{S}_{t}\}_{t\in\left[0,1\right]}$ is a local martingale and
therefore, the process $\{S_{t}\}_{t\in\left[0,1\right]}$ is a
semimartingale with the decomposition $
S_{t}=\hat{S}_{t}+\int_{0}^{t}\alpha_{u}\,d\left[S\right]_{u}$.
$\lozenge$
\end{proof}

\begin{corollary}[The Growth-Optimal Portfolio]
Under the assumptions of Theorem \ref{thm:main-Brownian}, the
wealth process $W^{\pi}$ is well-defined in the It\^o sense for all
$\pi\in\sH^2$ and
the stronger version of \eqref{equ:finite-log-Brown}
\begin{equation}
    \nonumber
    \begin{split}
        \sup_{\pi\in\sH^2} \E\left[\log(W^{\pi}_1)\right]<\infty
    \end{split}
\end{equation}
holds. Moreover, the supremum is attained by the process
$\alpha\in\sH^2$ from \eqref{decomposition}.
\end{corollary}
\begin{proof}
Since $\alpha\in\sH^2$, the process $\int_0^t \alpha_u \, d\hat{S}_u$ is a true
martingale and hence
$
\E\left[\log(W^\alpha_1)\right] = \frac12 \E\int_0^1\alpha^2_ud\left[S\right]_u.
$
The equation
\eqref{equ:Riesz} implies that
\begin{equation}%
\begin{split}
\nonumber
\E\left[\log(W^\pi_1)\right] &= {\mathbb{E}}\left[\int_{0}^{1}\pi_{u}\,dS_{u}\right]
-\frac{1}{2}\ {\mathbb{E}%
}\left[\int_{0}^{1}\pi_{u}^{2}\,d\left[S\right]_{u}\right]\\
&  ={\mathbb{E}}\left[\int_{0}^{1}\left(\pi_{u}\alpha_{u}-\frac{1}{2}\pi_{u}%
^{2}\right)\,d\left[S\right]_{u}\right]\leq\frac{1}{2}\ {\mathbb{E}}\left[\int_{0}^{1}\alpha_{u}%
^{2}\,d\left[S\right]_{u}\right] = \E\left[\log(W^\alpha_1)\right],
\end{split}
\label{equ:quad}%
\end{equation}
for any $\pi\in{\mathcal{H}}^{s}\cap\sH^{b}$. It suffices now to
use the density of ${\mathcal{H}}^{s}\cap\sH^{b}$ in
${\mathcal{H}}^{2}$. $\lozenge$
\end{proof}

\begin{remark}
A simple sufficient condition insuring that the local martingale
$\hat{S}$ in the semimartingale decomposition
\eqref{decomposition} is a true martingale
is that $\EE\left[\left[S\right]_1\right]<\infty$. In that case, $\hat{S}$ will be a
square-integrable martingale as well.
\end{remark}

\section{Processes With Jumps}

\subsection{Definition of the Wealth Process}

In this section we investigate whether it is possible to extend
the results of Theorem \ref{thm:main-Brownian} to the case when
the stochastic process $S$ admits jumps. We are facing the same
problem as in the previous section, i.e., the non-existence of the
canonical theory of stochastic integration for
non-semimartingales. However, with the motivation from Subsection
\ref{sub:canonical-definition},
a canonical definition of the stochastic exponential ${\mathcal{E}}%
(\pi\cdot S)$ of a process $\pi\cdot S$ for $\pi \in{\mathcal{H}}^{s}$
can be given:
\begin{align}
{\mathcal{E}}(\pi\cdot S)_{t}
 &\triangleq \nonumber\\
\exp&\left(  (\pi\cdot
S)_{t}-\frac{1}{2}\int_{0}^{t}\pi_{s}^{2}d\left[S\right]_{t}^{c}\right)
\prod_{s\leq t}(1+\pi_{s}\Delta S_{s})\exp(-\pi_{s}\Delta
S_{s})\label{equ:wealth-with-jumps},
\end{align}
provided that condition \eqref{equ:qv} - the existence of a finite
quadratic variation - holds. In the manner of Subsection
\ref{subs:fin-int}, the process
$W^{\pi}_t\triangleq{\mathcal{E}}(\pi\cdot S)_{t}$ can now be
interpreted as the evolution of the wealth of an investor who
invests the proportion $\pi_t$ of her/his total wealth at time $t$
in the risky asset.

\subsection{A Counterexample}

The goal of this subsection is to show that the results of Section
2.~cannot be extended to the class of processes with jumps, not
even in the case when the process $S$ is obtained from a
semimartingale via an enlargement of filtration. More precisely,
we construct two filtrations $\FFF\subseteq \GGG$, and an
$\FFF$-semimartingale $S$ with the following properties:
\begin{equation}
\left.
\begin{split}
&  \text{ (NS)\quad$S$ is \emph{not} a $\mathbb{G}$-semimartingale, but}\\
&  \text{ (FL)\quad$\sup_{\pi\in{\mathcal{H}}^{s}(\mathbb{G})}{\mathbb{E}%
}\left[\log(W_{1}^{\pi})\right]<\infty$}%
\end{split}
\right\}.  \label{equ:properties-in-example}%
\end{equation}

Before giving the details of our construction let us pause and try
to explain the intuition behind the example. The central idea is
that the introduction of jumps into the dynamics of the stock
price can lead to a drastic restriction of the set of portfolios
at the disposal of a logarithmic utility maximizer. Simply, any
portfolio leading to a negative terminal wealth with positive
probability yields an expected utility of negative infinity (as
usual, we set $\log(x)=-\infty$ for $x\leq 0$), and is, therefore,
clearly inferior to the constant portfolio $\pi\equiv0$. Suppose
that the process $S$ jumps in an unpredictable fashion, while its
continuous part fails the semimartingale property ``just barely''.
In that case, we are able to envision the situation in which the
non-semimartingality of $S$ cannot be exploited for unbounded
gains in logarithmic utility due to previously mentioned scarcity
of useful portfolio strategies. In other words, any strategy that
might lead to a large wealth suffers from the risk of finishing
negative with positive probability.

Theorem 7.2 in \cite{DelSch94} ensures the semimartingality of the
price process provided it is locally bounded
and satisfies the no free lunch with vanishing risk for
buy-and-hold strategies. Moreover, Example 7.5 (also in
\cite{DelSch94}) 
illustrates that the condition of local boundedness cannot be
relaxed. The idea of this example is similar to the above;
namely, the set of admissible portfolios may be almost empty.

Our construction of the process $S$ utilizes the following ingredients:

\begin{enumerate}
\item $B$ is a Brownian motion and $\mathbb{F}%
^{B}\triangleq\{{\mathcal{F}}_{t}^{B}\}_{t\in\left[0,1\right]}$ is the
(right-continuous and complete) augmentation of the filtration
generated by $B$.

\item $M$ is the Gaussian martingale given by
$M_{t}\triangleq\int_{0}^{t}\sigma(u)\,dB_{u}$, where
\[
\sigma(t)\triangleq\frac{\left\vert \log(1-t)\right\vert ^{-2/3}}{\sqrt{1-t}%
}{\mathbf{1}}_{\left\{  1>t>\frac{1}{2}\right\}  }.
\]

\item $N^{1}$ and $N^{2}$
are two independent Poisson processes (of course also independent
of the Brownian motion $B$).

\item $N$ is the pure-jump process
defined by $N_{t}\triangleq N_{t}^{1}-N_{t}^{2}$
and $\mathbb{F}^{N}$ is the filtration generated by the process $N$ (or,
equivalently, by $N^{1}$ and $N^{2}$).
\end{enumerate}
Having introduced the necessary ingredients, the process $S$, announced
in \eqref{equ:properties-in-example}, is defined by
\[
S_{t}\triangleq M_{t}+\int_{0}^{t}\frac{1}{1-u}\,dN_{u},\
t\in\left[0,1\right].
\]
$S$ is clearly an $\mathbb{F}$-semimartingale, where $\mathbb{F}$ is the
filtration generated by $B$ and $N$, i.e. $\mathbb{F}\triangleq\mathbb{F}%
^{B}\vee\mathbb{F}^{N}$. Let the enlarged filtration $\mathbb{G}$ be defined
by adding the information about the terminal value $B_{1}$ of the Brownian
motion $B$ to $\mathbb{F}$, i.e. ${\mathcal{G}}_{t}\triangleq{\mathcal{F}}_{t}%
\vee\sigma(B_{1})$, $t\in \left[0,1\right]$. The properties (NS)
and (FL) in (\ref{equ:properties-in-example}), are now established
through the following lemmas.

\begin{lemma}
Property (NS) in (\ref{equ:properties-in-example}) holds true:
$S$ is \emph{not} a $\mathbb{G}$-semimartingale.
\end{lemma}

\begin{proof}
It is enough to show that $\{M_{t}\}_{t\in\left[0,1\right]}$ is not a $\mathbb{G}%
$-semimartingale. This is, however, exactly the content of Theorem
IV.7 in \cite{Pro04} and the example following it. $\lozenge$
\end{proof}

\begin{lemma}
\label{pro:is-constrained} Let $\pi\in{\mathcal{H}}^{s}(\mathbb{G})$ be a
simple integrand and let $W^{\pi}$ be the corresponding wealth process, as
defined in (\ref{equ:wealth-with-jumps}). If ${\mathbb{P}}\left[W_{1}^{\pi}>0\right]=1$
then
\[
\pi_{t}\in\big(-(1-t),1-t\,\big),\
(\lambda\otimes{\mathbb{P})}\text{-a.e.},
\]
where $\lambda$ denotes the Lebesgue measure on $\left[0,1\right]$.
\end{lemma}
Before proving Lemma \ref{pro:is-constrained}, we require the following result:
\begin{lemma}
\label{lem:difference-of-Poissons} Let $N$ be a
difference of two independent $\mathbb{G}$-Poisson processes, and let
$\beta$ be a $\mathbb{G}$-predictable process taking
values in the set $\left\{  -1,1\right\}  $. Then the process $N$,
defined by the integral $\tilde{N}_{t}\triangleq\int_{0}%
^{t}\beta_{s}\,dN_{s}$ can be decomposed into a difference of two independent
$\mathbb{G}$-Poisson processes.
\end{lemma}

\begin{proof}
Let $N_{t}\triangleq N_{t}^{+}-N_{t}^{-}$ be the decomposition of
$N$ into two independent Poisson processes, and let
$\beta_{t}^{+}\triangleq \max(\beta_{t},0)$ and
$\beta_{t}^{-}\triangleq \max(-\beta_{t},0)$ so that
$\beta_{t}=\beta_{t}^{+}-\beta _{t}^{-}$ and
$\beta_{t}^{+}+\beta_{t}^{-}=1$, for all $t\in\left[0,1\right]$, a.s. The
processes $\tilde{N}^{+}$ and $\tilde{N}^{-}$ defined by
\begin{equation}%
\begin{split}
\tilde{N}_{t}^{+}  &  \triangleq\int_{0}^{t}\beta_{s}^{+}\,dN_{s}^{+}+\int_{0}^{t}%
\beta_{s}^{-}\,dN_{s}^{-},\ \text{and}\\
\tilde{N}_{t}^{-}  &  \triangleq\int_{0}^{t}\beta_{s}^{-}\,dN_{s}^{+}+\int_{0}^{t}%
\beta_{s}^{+}\,dN_{s}^{-}.
\end{split}
\nonumber
\end{equation}
have the following properties

\begin{enumerate}
\item $\tilde{N}^{+}$ and $\tilde{N}^{-}$ are non-decreasing processes and
increase only by jumps of magnitude $1$.

\item $\tilde{N}^{+}_{t}-(\beta^{+}_{t}+\beta^{-}_{t})t=\tilde{N}^{+}_{t}-t$
and $\tilde{N}^{-}_{t}-(\beta^{+}_{t}+\beta^{-}_{t})t=\tilde{N}^{-}_{t}-t$ are martingales.

\item The intersection of the sets of jump-times for $\tilde{N}^{+}$ and
$\tilde{N}^{-}$ is empty, a.s.
\end{enumerate}

Items (1) and (2) imply that $\tilde{N}^{+}$ and $\tilde{N}^{-}$
are $\mathbb{G}$-Poisson processes and (3) is enough to conclude
that they are independent (see \cite{Bre81}). Therefore,
$\tilde{N}=\tilde{N}^{+}-\tilde {N}^{-}$ is a difference of two
Poisson processes. $\lozenge$
\end{proof}

\begin{proof}
(Of Lemma \ref{pro:is-constrained}) Let the process $\hat{\pi}$ be
defined as $\hat{\pi}_{t}\triangleq\pi_{t}/(1-t)1_{\{t<1\}}$, and
suppose that the predictable set
$A\triangleq\left\{  (t,\omega)\in\left[0,1\right]\times\Omega\,:\,\left\vert \hat{\pi}%
_{t}(\omega)\right\vert \geq1\right\}  $ satisfies $(\lambda\otimes
{\mathbb{P}})\left[A\right]>0$. The expression (\ref{equ:wealth-with-jumps})
for the wealth $W_{1}^{\pi}$ can be split into two factors, one of
which is an exponential and the other is the product of the
form
\[
Y\triangleq\prod_{s\leq1}(1+\pi_{s}\frac{1}{1-s}\Delta N_{s})=\prod_{s\leq
1}(1+\hat{\pi}_{s}\Delta N_{s}).
\]
The sign of $W_{1}^{\pi}$ is equal to the sign of $Y$, so in order to reach a
contradiction, it will be enough to prove that ${\mathbb{P}}\left[Y\leq0\right]>0$.

Define the process $\tilde{N}_{t}\triangleq\int_{0}^{t}\operatorname{sgn}%
(\hat{\pi}_{s})\,dN_{s}$, where $\operatorname{sgn}(x)=1$ for $x\geq0$ and
$\operatorname{sgn}(x)=-1$, otherwise. By Lemma
\ref{lem:difference-of-Poissons}, there exist two independent Poisson
processes $\tilde{N}^{+}$ and $\tilde{N}^{-}$ such that $\tilde{N}=\tilde
{N}^{+}-\tilde{N}^{-}$, and
\[
Y=\prod_{s\leq1}(1+\left\vert \hat{\pi}_{s}\right\vert \Delta\tilde{N}_{s}).
\]
Let $J$ be the event that $\tilde{N}^{-}$ jumps exactly once on the
set $A$, i.e.,
$$J\triangleq\left\{
\int_{0}^{1}{\mathbf{1}}_{{A}}(s)\,d\tilde {N}^{-}_s=1\right\}.$$
Since $\left\vert \hat{\pi}_{s}\right\vert \geq1$ on $A$, it is easy
to see that
${\mathbb{P}}\left[Y\leq0\right]\geq{\mathbb{P}}\left[J\right]$. In
order to show that ${\mathbb{P}}\left[J\right]>0$, we first define $J^{\prime}%
\triangleq\left\{  \int_{0}^{1}{\mathbf{1}}_{{A}}(s)\,d\tilde{N}^{-}_s%
\geq1\right\}  \supseteq J$. The martingale property of the process
$X_{t}\triangleq\int_{0}^{t}{\mathbf{1}}_{{A}}(s)\,d\tilde{N}_{s}^{-}-\int_{0}%
^{t}{\mathbf{1}}_{{A}}(s)\,ds$ implies that
\[
{\mathbb{E}}\left[\int_{0}^{1}{\mathbf{1}}_{{A}}(s)\,d\tilde{N}_{s}^{-}%
\right]={\mathbb{E}}\left[\int_{0}^{1}{\mathbf{1}}_{{A}}(s)\,ds\right]=(\lambda\otimes
{\mathbb{P}})\left[A\right]>0,
\]
showing that the ${\mathbb{N}}\cup\left\{  0\right\}  $-valued random variable
$\int_{0}^{t}{\mathbf{1}}_{{A}}(s)\,d\tilde{N}_{s}^{-}$ has a strictly
positive expectation, and thus ${\mathbb{P}}\left[J^{\prime}\right]>0$. Define $\tau^{1}$
to be the first jump time of the process $\tilde{N}^{-}$. By the $\mathbb{G}%
$-L\'{e}vy property of the Poisson process $\tilde{N}^-$, the
process $\hat {N}_{t}\triangleq
\tilde{N}_{\tau_{1}+t}^{-}-\tilde{N}_{\tau_{1}}^{-}$ is a Poisson
process, independent of ${\mathcal{G}}_{\tau_{1}}$. The
probability that $\hat{N}$ will stay constant for one unit of time
is strictly positive, and, consequently, so is the probability
that $\tilde{N}^{-}$ will jump exactly once on $A$. This implies
that $\mathbb{P}\left[Y\mathbb{\leq}0\right]>0$ - a contradiction.
$\lozenge$
\end{proof}

\begin{proposition}
There exists a constant $C<\infty$ such that
\begin{equation}
{\mathbb{E}}\left[\log(W_{1}^{\pi})\right]\leq C,\ \text{for all $\pi\in{\mathcal{H}%
}^{s}(\mathbb{G})$}. \label{equ:bounded-log}%
\end{equation}

\end{proposition}

\begin{proof}
By Lemma \ref{pro:is-constrained}, it is enough to show that
(\ref{equ:bounded-log}) is true for all $\pi\in{\mathcal{H}}^{s}%
(\mathbb{G})$, with the additional property that $\left\vert \pi
_{s}\right\vert <(1-s)$, $\lambda\otimes{\mathbb{P}}$-a.e.

The expression for $W_{1}^{\pi}$ given in (\ref{equ:wealth-with-jumps})
factorizes into an exponential and a product of transformed jumps, so that
${\mathbb{E}}\left[\log(W_{1}^{\pi})\right]\leq C(\pi)+J(\pi)$, where%
\begin{equation}%
\begin{split}
C(\pi) &  \triangleq{\mathbb{E}}\left[\log({\mathcal{E}}(\pi\cdot M))_{1}%
\right]\quad\text{ and}\\
D(\pi) &  \triangleq{\mathbb{E}}
\left[\sum_{s\leq1}\log(1+\frac{\pi_{s}}{1-s}\Delta
N_{s})\right]\leq\mathbb{E}\left[\sum_{s\leq1}\frac{\pi_{s}}{1-s}\Delta N_{s}\right]=0.
\end{split}
\nonumber
\end{equation}
To obtain a bound on $C(\pi)$ we first apply Jensen's inequality and then
Fatou's Lemma to obtain
\[
C(\pi)\leq\log({\mathbb{E}}\left[{\mathcal{E}}(\pi\cdot M)_{1}\right])
\leq\liminf_{t\rightarrow1}\log({\mathbb{E}}
\left[{\mathcal{E}}(\pi\cdot M)_{t}\right]).
\]
Now, all we need is a uniform bound (in $\pi$ and $t$) on ${\mathbb{E}%
}\left[{\mathcal{E}}(\pi\cdot M)_{t}\right]$, for $t<1$. This is accomplished by noting
that the process $M$ is a $\mathbb{G}$-semimartingale on any interval
$\left[0,u\right],u<1$, with the semimartingale decomposition $M=\hat{M}+(M-\hat{M})$,
where the $\mathbb{G}$-martingale $\hat{M}$ is given by :
\[
\hat{M}_{t}\triangleq\int_{0}^{t}\sigma(u)\left(  dB_{u}-\frac{B_{1}-B_{u}%
}{1-u}du\right)  .
\]
This allows us to write
\begin{align*}
{\mathcal{E}}(\pi\cdot M)_{t} &  ={\mathcal{E}}\left(  \int_{0}^{t}\pi
_{u}\,d\hat{M}_{u}+\int_{0}^{t}\pi_{u}\sigma(u)\left(  \frac{B_{1}-B_{u}}%
{1-u}\right)  \,du\right)  \\
&  =\exp\left(  (\pi\cdot\hat{M})_{t}-(\pi^{2}\cdot\left[\hat{M}%
\right])_{t}\right)  \\
&  \times\exp\left(  \frac{1}{2}\int_{0}^{t}\pi_{u}^{2}\sigma(u)^{2}%
du+\int_{0}^{t}\pi_{u}\sigma(u)\left(  \frac{B_{1}-B_{u}}{1-u}\right)
du\right)  .
\end{align*}
The Cauchy-Schwartz inequality, combined with the observation that the square of the
exponential $\exp\left(  (\pi\cdot\hat{M})_{t}-(\pi^{2}\cdot\left[\hat
{M}\right])_{t}\right)  $ is a positive local martingale and, hence, a
supermartingale, yields:
\[
{\mathbb{E}}\left[{\mathcal{E}}(\pi\cdot
  M)_{t}\right]^{2}\leq{\mathbb{E}}
\left[
\exp\left(  \int_{0}^{t}\pi_{u}^{2}\sigma(u)^{2}du+2\int_{0}^{t}\pi_{u}%
\sigma(u)\left(  \frac{B_{1}-B_{u}}{1-u}\right)  du\right)  \right]  .
\]
To see that this expectation can be bounded away from $\infty$,
independently of $t$ and $\pi$, we can use the bound
$|\pi_{t}|\leq1-t$, the explicit form of the function $\sigma$,
and the fact that all exponential moments of the random variable
$\sup_{t\in\left[0,1\right]}|B_{t}|$ are finite. $\lozenge$
\end{proof}

\begin{acknowledgements}
We thank Dmitry Kramkov and Morten Mosegaard Christensen as well
as anonymous referees for useful comments.
\end{acknowledgements}

\end{document}